\documentclass[aps,preprint,showpacs,groupedaddress]{revtex4-1}
\usepackage{amssymb}
\usepackage{mathrsfs}
\usepackage{amsmath}
\usepackage{mathtools}
\usepackage{pstricks}
\usepackage{color}
\usepackage{graphicx}
\usepackage{amsthm}
\newtheorem*{theorem}{Theorem}
\usepackage{physics}

\begin{document}


\title{\bf Asymptotic single-particle states and exact Lorentz-violating all-loop quantum corrections for scalar field theory}



\author{P. R. S. Carvalho}
\email{prscarvalho@ufpi.edu.br}
\affiliation{\it Departamento de F\'\i sica, Universidade Federal do Piau\'\i, 64049-550, Teresina, PI, Brazil}

\author{M. I. Sena-Junior}
\email{marconesena@poli.br}
\affiliation{\it Escola Polit\'{e}cnica de Pernambuco, Universidade de Pernambuco, 50720-001, Recife, PE, Brazil}
\affiliation{\it Instituto de F\'{i}sica, Universidade Federal de Alagoas, 57072-900, Macei\'{o}, AL, Brazil}




\begin{abstract}
We perform the all-loop renormalization of the O($N$) $\lambda\phi^{4}$ scalar field theory with Lorentz violation which is exact in the Lorentz-violating $K_{\mu\nu}$ coefficients. This task is fulfilled analytically firstly explicitly at next-to-leading order and later at all loop levels through an induction process based on a general theorem emerging from the exact approach. We show that the exact results reduce to the approximated ones, previously obtained, in the appropriated limit and comment on their implications. The current exact calculation involving such a symmetry breaking mechanism in the referred theory is the first one in literature for our knowledge. At the end, we analyze the effect of Lorentz violation on the asymptotic single-particle states.
\end{abstract}


\maketitle


\section{Introduction}\label{Introduction} 

\par The possibility of new physical effects due to Lorentz symmetry breaking mechanisms has recently attracted the attention of both high energy physics \cite{PhysRevD.92.045016,0004-637X-806-2-269,doi:10.1142/S0217751X15500724,vanTilburg2015236,PhysRevD.86.125015,PhysRevD.84.065030,Carvalho2013850,Carvalho2014320} and condensed matter theory communities \cite{EurophysLett10821001,doi:10.1142/S0217979215502598,doi:10.1142/S0219887816500493}. The most of these studies are, from a practical and technical viewpoint, approached through finite perturbative expansions in some characteristic Lorentz violating (LV) parameter, usually at its first order \cite{PhysRevD.92.045016,0004-637X-806-2-269,doi:10.1142/S0217751X15500724,vanTilburg2015236,PhysRevD.86.125015,Carvalho2014320,PhysRevD.58.116002,PhysRevD.65.056006,PhysRevD.79.125019,PhysRevD.77.085006}. Just a few of them displayed results, tediously computed, at its second order \cite{PhysRevD.84.065030,Carvalho2013850,EurophysLett10821001,doi:10.1142/S0217979215502598,doi:10.1142/S0219887816500493} and only one, up to the moment for our knowledge, exactly \cite{PhysRevLett.83.2518}. 

\par The theory studied in this Letter, a LV one characterized by the LV $K_{\mu\nu}$ coefficients, is specified by three parameters: field, mass and coupling constant. This theory is represented by the CPT-even scalar sector of the so called standard model extension \cite{PhysRevD.58.116002}, the LV version of the conventional standard model of elementary particles and fields \cite{Itzykson,Peskin}. The referred symmetry breaking mechanism is due to a LV kinetic term of the form $K_{\mu\nu}\,\partial^{\mu}\phi\,\partial^{\nu}\phi$. This is not the only LV kinetic term preserving the O($N$) symmetry. In fact, higher-derivative ones of the form $\frac{1}{2\Lambda_{L}^{2n-2}}(\partial^{n}\phi)^{2}$, where $\Lambda_{L}$ is an energy scale and $n > 1$, were also proposed in literature and will not be discussed here \cite{PhysRevD.76.125011}. Unfortunately, this theory is plagued by ultraviolet divergences and a theory making physical sense is attained only if its finite version can be found with the renormalization of its parameters or, equivalently, the computation of the $\beta$-function $\beta(g)$ and field $\gamma(g)$ and mass $\gamma_{m}(g)$ anomalous dimensions, where $g$ is the dimensionless renormalized coupling constant. In fact, a finite theory, although in a limited range of small values of the LV coefficients, was found. The $\beta$-function and $\gamma(g)$ \cite{PhysRevD.84.065030} and $\gamma_{m}(g)$ \cite{Carvalho2013850} were renormalized at two-loop level and second order in $K_{\mu\nu}$. Furthermore, the computation of $\gamma(g)$ has been improved up to three-loop level and at first order in the LV coefficients \cite{Carvalho2014320}. 
In the last three works just mentioned, an attempt for obtaining all-loop expressions for $\beta(g)$, $\gamma(g)$ and $\gamma_{m}(g)$ was proposed, although being valid only for small LV coefficients values, as an induction process starting from their expressions for finite second order in $K_{\mu\nu}$ through coordinates redefinition techniques or equivalently the joint application of a redefinition of the field and coupling constant parameters. We show, as a consequence of the exact calculation, that an exact version of such a goal can be rigorously achieved and interpreted in a simpler and clearer way, than in the earlier non-exact one, by regarding just a single concept of the theory: the loop level considered. Although the earlier approximated results, tediously computed as an expansion in $K_{\mu\nu}$, can give a good insight on the effect of taking into account LV loop quantum corrections to the theory, a truly reliable and complete understanding of the final renormalized theory can be achieved only through an exact treatment of its LV radiative quantum corrections valid for the entire finite range of $K_{\mu\nu}$ values and loop levels. This is the main aim of this Letter. Moreover, the exact treatment presented here could inspire further works in considering exact LV effects in systems described by others LV theories.

\par In this Letter we compute analytically the exact LV quantum corrections for O($N$) self-interacting scalar field theories. Firstly, the LV quantum corrections are evaluated explicitly at next-to-leading loop order and later at any loop level. At the conclusions, we comment on the implications of this exact calculation.

\section{Exact Lorentz-violating next-to-leading loop order quantum corrections}

\par The theory to be renormalized is described by the Lagrangian density
\begin{eqnarray}\label{bare Lagrangian density}
\mathcal{L}_{B} = \frac{1}{2}\,\partial_{\mu}\phi_{B}\,\partial^{\mu}\phi_{B} + K_{\mu\nu}\,\partial^{\mu}\phi_{B}\,\partial^{\nu}\phi_{B} + \frac{1}{2}\,m_{B}^{2}\,\phi_{B}^{2} + \frac{\lambda_{B}}{4!}\phi_{B}^{4}, 
\end{eqnarray}
defined in $d$-dimensional Euclidean spacetime (we can also define the referred Lagrangian density in Minkowski spacetime through a Wick rotation) where the bare field $\phi_{B}$, coupling constant $\lambda_{B}$ and mass $m_{B}$ are ultraviolet divergent and the LV coefficients $K_{\mu\nu}$ are dimensionless, symmetric and equal for all $N$ components of the field such that the O($N$) symmetry of the field remains intact. We treat the LV $K_{\mu\nu}$ coefficients exactly, thus they can assume any finite values. We have to cure the ultraviolet divergences of the theory by applying the most general and elegant renormalization method for massive theories, the Bogoliubov-Parasyuk-Hepp-Zimmermann (BPHZ) one \cite{BogoliubovParasyuk,Hepp,Zimmermann}. We follow the definitions and notation of Ref. \cite{Kleinert}. From the initially bare Lagrangian density we can obtain the divergent primitively one-particle-irreducible ($1$PI) vertex parts, the two- and four-point functions, namely $\Gamma_{B}^{(2)}$ and $\Gamma_{B}^{(4)}$, respectively. The ultraviolet divergences of the theory are contained on these correlation functions around $d = 4$ and if we dimensionally regularize them in $d = 4 - \epsilon$, we can attain the finite final theory, thus obtaining the corresponding renormalized parameters $\phi$, $\lambda$ and $m$. Any higher vertex parts can be written in terms of the primitively ones as an skeleton expansion \cite{ZinnJustin} and the renormalization of the latter ensures the renormalization of the former. We can eliminate the divergences of the bare theory by adding terms to the bare $1$PI vertex parts at a given loop order. These terms are called counterterms. The countertems can be viewed as being originated by new terms to be included in the initially bare Lagrangian density at the corresponding order. These new terms have the same functional form of the ones in the initial Lagrangian density. The new and the initial terms can be combined such that the resulting terms, which are now finite, are proportional to the initial ones but now scaled by constants. These constants are called renormalization constants and absorb the ultraviolet divergences of the theory. They are defined by $\phi = Z_{\phi}^{-1/2}\phi_{B}$,$g = \mu^{-\epsilon}Z_{\phi}^{2}\lambda_{B}/Z_{g}$ and $m^{2} = Z_{\phi}m_{B}^{2}/Z_{m^{2}}$, where $g=\mu^{-\epsilon}\lambda$ is the renormalized dimensionless coupling constant and $\mu$ is an arbitrary mass parameter. Thus the $1$PI vertex parts at that order are finite. We can proceed at the next order and so on order by order in perturbation theory for attaining the final renormalized vertex parts and thus the renormalized Lagrangian density now depending on the renormalized parameters. The renormalized vertex parts satisfy to the Callan-Symanzik equation
\begin{eqnarray}
\left[\mu\frac{\partial}{\partial\mu} + \beta(g)\,\frac{\partial}{\partial g} - n\,\gamma(g) + \gamma_{m}(g)\,m\frac{\partial}{\partial m}\right] \Gamma^{(n)}(P_{1},...,P_{n};m,g,\mu) = 0
\end{eqnarray}
where the $\beta$-function and field $\gamma$ and mass $\gamma_{m}$ anomalous dimensions are given by
\begin{eqnarray}
\beta(g) = \mu\frac{\partial g}{\partial\mu}\Bigg\vert_{B}, \hspace{1cm}\gamma(g) = \frac{1}{2}\mu\frac{\partial}{\partial\mu}Z_{\phi}\bigg\vert_{B}, \hspace{1cm}\gamma_{m}(g) = \frac{\mu}{m}\frac{\partial m}{\partial\mu}\Bigg\vert_{B}
\end{eqnarray}
and $|_{B}$ means that the respective parameters must be that of bare theory. At next-to-leading loop level, the renormalization constants are written as a Laurent expansion in $\epsilon$ by
\begin{eqnarray}
&& Z_{\phi}(g,\epsilon^{-1}) = 1 +   \frac{1}{P^2 + K_{\mu\nu}P^{\mu}P^{\nu}} \Biggl[ \frac{1}{6} \mathcal{K} 
\left(\parbox{10mm}{\includegraphics[scale=1.0]{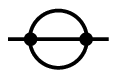}}
\right) \Biggr|_{m^{2}=0} S_{\parbox{10mm}{\includegraphics[scale=0.5]{fig6.eps}}} +   \frac{1}{4} \mathcal{K} 
\left(\parbox{10mm}{\includegraphics[scale=1.0]{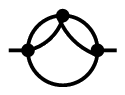}} \right) \Biggr|_{m^{2}=0}S_{\parbox{8mm}{\includegraphics[scale=0.5]{fig7.eps}}} + \nonumber \\&&  \frac{1}{3} \mathcal{K}
  \left(\parbox{10mm}{\includegraphics[scale=1.0]{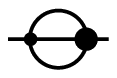}} \right) S_{\parbox{6mm}{\includegraphics[scale=0.5]{fig26.eps}}} \Biggr],
\end{eqnarray}
\begin{eqnarray}
&& Z_{g}(g,\epsilon^{-1}) = 1 +  \frac{1}{\mu^{\epsilon}g} \Biggl[ \frac{1}{2} \mathcal{K} 
\left(\parbox{10mm}{\includegraphics[scale=1.0]{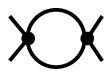}} + 2 \hspace{1mm} perm.
\right) S_{\parbox{10mm}{\includegraphics[scale=0.5]{fig10.eps}}} +  \frac{1}{4} \mathcal{K} 
\left(\parbox{15mm}{\includegraphics[scale=1.0]{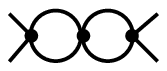}} + 2 \hspace{1mm} perm. \right) S_{\parbox{10mm}{\includegraphics[scale=0.5]{fig11.eps}}}  + \nonumber \\&& \frac{1}{2} \mathcal{K} 
\left(\parbox{12mm}{\includegraphics[scale=.8]{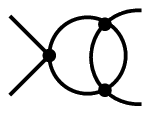}} + 5 \hspace{1mm} perm. \right) S_{\parbox{10mm}{\includegraphics[scale=0.4]{fig21.eps}}} +  \frac{1}{2} \mathcal{K} 
\left(\parbox{10mm}{\includegraphics[scale=1.0]{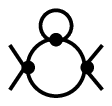}} + 2 \hspace{1mm} perm.
\right) S_{\parbox{10mm}{\includegraphics[scale=0.5]{fig13.eps}}} +  \nonumber \\&& \mathcal{K}
  \left(\parbox{10mm}{\includegraphics[scale=1.0]{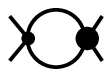}} + 2 \hspace{1mm} perm. \right) S_{\parbox{6mm}{\includegraphics[scale=0.5]{fig25.eps}}} + \mathcal{K}
  \left(\parbox{10mm}{\includegraphics[scale=1.0]{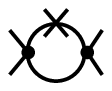}} + 2 \hspace{1mm} perm. \right) S_{\parbox{6mm}{\includegraphics[scale=0.5]{fig24.eps}}}\Biggr],
\end{eqnarray}
\begin{eqnarray}
&& Z_{m^{2}}(g,\epsilon^{-1}) = 1 + \frac{1}{m^{2}} \Biggl[ \frac{1}{2} \mathcal{K}
\left( \parbox{10mm}{\includegraphics[scale=1.0]{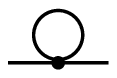}} \right)S_{\scalebox{0.3}{\parbox{10mm}{\includegraphics[scale=1.0]{fig1.eps}}}} +  \frac{1}{4} \mathcal{K}  \left( \parbox{10mm}{\includegraphics[scale=1.0]{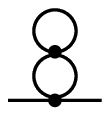}} \right)S_{\scalebox{0.3}{\parbox{10mm}{\includegraphics[scale=1.0]{fig2.eps}}}}
   +  \frac{1}{2} \mathcal{K}
  \left( \parbox{10mm}{\includegraphics[scale=1.0]{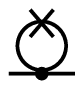}} \right)S_{\scalebox{0.3}{\parbox{10mm}{\includegraphics[scale=1.0]{fig22.eps}}}} 
  + \nonumber \\&&  \frac{1}{2} \mathcal{K} \left( \parbox{10mm}{\includegraphics[scale=1.0]{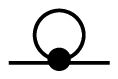}} \right)S_{\scalebox{0.3}{\parbox{10mm}{\includegraphics[scale=1.0]{fig23.eps}}}} +  \frac{1}{6} \mathcal{K}
  \left( \parbox{10mm}{\includegraphics[scale=1.0]{fig6.eps}}\hspace*{0.2cm} \right) \Biggr|_{P^2 + K_{\mu\nu}P^{\mu}P^{\nu} = 0} S_{\scalebox{0.3}{\parbox{10mm}{\includegraphics[scale=1.0]{fig6.eps}}}} \Biggr],
\end{eqnarray}
where the operator $\mathcal{K}$ selects only the divergent terms of the referred diagram. The factor $S_{\parbox{3mm}{\includegraphics[scale=.3]{fig1.eps}}}$ is the symmetry factor for the corresponding diagram in a theory where the field has $N$-components. In the fourteen Feynman diagrams above, the internal lines represents the LV renormalized free propagator given by $q^{2} + K_{\mu\nu}q^{\mu}q^{\nu} + m^{2}$. For evaluating the diagrams exactly in the LV parameters and therefore for any of its finite values, thus avoiding to expand them in the limited range of small powers of  $K_{\mu\nu}$ and to turn out the calculations as tedious ones, we observe that the free propagator has a quadratic and bilinear dependence on its momentum. This permits us to write in the $d$-dimensional space of momentum integrals that $q^{2} + K_{\mu\nu}q^{\mu}q^{\nu} + m^{2} = (\delta_{\mu\nu} + K_{\mu\nu})q^{\mu}q^{\nu} + m^{2} \equiv q^{t}(\mathbb{I}$ + $\mathbb{K})q + m^{2}$, where in last expression $q$ is a $d$-dimensional vector represented by a column matrix and is $q^{t}$ represented by a row matrix, and $\mathbb{I}$ and $\mathbb{K}$ are the matrix representations of the identity and $K_{\mu\nu}$, respectively. Now making the substitution $q^{\prime} = \sqrt{\mathbb{I} + \mathbb{K}}\hspace{1mm}q$, the exact theory acquires LV contributions in two forms. One of them comes from the volume elements of Feynman diagrams such that $d^{d}q^{\prime} = \sqrt{\det(\mathbb{I} + \mathbb{K})}\,d^{d}q$, thus $d^{d}q = d^{d}q^{\prime}/\sqrt{\det(\mathbb{I} + \mathbb{K})}$.  Defining this full or exact LV contribution as $\mathbf{\Pi} \equiv 1/\sqrt{\det(\mathbb{I} + \mathbb{K})}$, we see that for slight violations of Lorentz symmetry, we have that the full LV $\mathbf{\Pi}$ factor reduces to its perturbative second order counterpart $\Pi \simeq \Pi^{(0)} + \Pi^{(1)} + \Pi^{(2)}$, where $\Pi^{(i)}$ is the LV contribution of order $i$ in $K_{\mu\nu}$ \cite{PhysRevD.84.065030,Carvalho2013850,Carvalho2014320}. The theory is also mo\-di\-fied by another type of LV contribution, that involving external momenta. It is shown in the expressions for the  momentum-dependent Feynman diagrams evaluated in dimensional regularization 
\begin{eqnarray}
&& \int \frac{d^{d}q}{(2\pi)^{d}} \frac{1}{(q^{2} + 2Pq + M^{2})^{\alpha}} =  \hat{S}_{d}\frac{1}{2}\frac{\Gamma(d/2)}{\Gamma(\alpha)}\frac{\Gamma(\alpha - d/2)}{(M^{2} - P^{2})^{\alpha - d/2}},   
\end{eqnarray}
$\hat{S}_{d}=S_{d}/(2\pi)^{d}=2/(4\pi)^{d/2}\Gamma{(d/2)}$, and $S_{d}=2\pi^{d/2}/\Gamma(d/2)$ is the surface area of a unit $d$-dimensional sphere. Its finite value in four-dimensional spacetime is  $\hat{S}_{4}=2/(4\pi)^{2}$. This definition adds to each loop integration a factor of $\hat{S}_{4}$ at four-dimensional spacetime and avoids the appearance of Euler-Mascheroni constants in the middle of calculations \cite{Amit}. Now making $q^{\prime} \rightarrow P^{\prime}$ and $q \rightarrow P$, $P^{\prime 2} = P^{2} + K_{\mu\nu}P^{\mu}P^{\nu}$. The computed diagrams are displayed below  
\begin{eqnarray}
\left(\parbox{12mm}{\includegraphics[scale=1.0]{fig6.eps}}\right)\Biggr|_{m^{2}=0} =  -\frac{g^{2}(P^2 + K_{\mu\nu}P^{\mu}P^{\nu})}{2(4\pi)^{4}\epsilon} \left[ 1 + \frac{1}{4}\epsilon -2\epsilon\, J_{3}(P^{2} + K_{\mu\nu}P^{\mu}P^{\nu}) \right]\mathbf{\Pi}^{2},
\end{eqnarray}
\begin{eqnarray}
\parbox{12mm}{\includegraphics[scale=1.0]{fig7.eps}}\bigg|_{m^{2}=0} = \frac{4(P^2 + K_{\mu\nu}P^{\mu}P^{\nu})g^{3}}{3(4\pi)^{6}\epsilon^{2}} \left[1 + \frac{1}{2}\epsilon - 3\epsilon\, J_{3}(P^{2} + K_{\mu\nu}P^{\mu}P^{\nu})\right]\mathbf{\Pi}^{3},
\end{eqnarray}
\begin{eqnarray}
\parbox{10mm}{\includegraphics[scale=1.0]{fig26.eps}} \quad = -\frac{3(P^{2} + K_{\mu\nu}P^{\mu}P^{\nu})g^{3}}{2(4\pi)^{6}\epsilon^{2}}\left[1 + \frac{1}{4}\epsilon - 2\epsilon\, J_{3}(P^{2} + K_{\mu\nu}P^{\mu}P^{\nu})\right]\mathbf{\Pi}^{3},
\end{eqnarray}
\begin{eqnarray}
\parbox{10mm}{\includegraphics[scale=1.0]{fig10.eps}} =  \frac{2\mu^{\epsilon}g^{2}}{(4\pi)\epsilon} \left[1 - \frac{1}{2}\epsilon - \frac{1}{2}\epsilon J(P^{2} + K_{\mu\nu}P^{\mu}P^{\nu}) \right]\mathbf{\Pi},
\end{eqnarray}
\begin{eqnarray}
\parbox{16mm}{\includegraphics[scale=1.0]{fig11.eps}} = -\frac{4\mu^{\epsilon}g^{3}}{(4\pi)^{4}\epsilon^{2}} \left[1 - \epsilon - \epsilon J(P^{2} + K_{\mu\nu}P^{\mu}P^{\nu}) \right]\mathbf{\Pi}^{2},\quad\quad
\end{eqnarray}
\begin{eqnarray}
\parbox{12mm}{\includegraphics[scale=0.8]{fig21.eps}} = -\frac{2\mu^{\epsilon}g^{3}}{(4\pi)^{4}\epsilon^{2}} \left[1 - \frac{1}{2}\epsilon - \epsilon J(P^{2} + K_{\mu\nu}P^{\mu}P^{\nu}) \right]\mathbf{\Pi}^2, \quad\quad
\end{eqnarray}
\begin{eqnarray}
\parbox{12mm}{\includegraphics[scale=1.0]{fig13.eps}} =   \frac{2\mu^{\epsilon}g^{3}}{(4\pi)^{4}\epsilon^{2}} J_{4}(P^{2} + K_{\mu\nu}P^{\mu}P^{\nu})\mathbf{\Pi}^{2},
\end{eqnarray}
\begin{eqnarray}
\parbox{10mm}{\includegraphics[scale=1.0]{fig25.eps}} = \frac{6\mu^{\epsilon}g^{3}}{(4\pi)^{4}\epsilon^{2}} \left[1 - \frac{1}{2}\epsilon - \frac{1}{2}\epsilon J(P^{2} + K_{\mu\nu}P^{\mu}P^{\nu}) \right]\mathbf{\Pi}^{2},\quad\quad
\end{eqnarray}
\begin{eqnarray}
\parbox{12mm}{\includegraphics[scale=1.0]{fig24.eps}} =  -\frac{2\mu^{\epsilon}g^{3}}{(4\pi)^{4}\epsilon^{2}} J_{4}(P^{2} + K_{\mu\nu}P^{\mu}P^{\nu})\,\mathbf{\Pi}^{2},
\end{eqnarray}
\begin{eqnarray}
\parbox{12mm}{\includegraphics[scale=1.0]{fig1.eps}} =
\frac{2m^{2}g}{(4\pi)^{2}\epsilon}\left[ 1 - \frac{1}{2}\epsilon\ln\left(\frac{m^{2}}{4\pi\mu^{2}}\right)\right]\mathbf{\Pi},
\end{eqnarray}
\begin{eqnarray}
\parbox{10mm}{\includegraphics[scale=1.0]{fig2.eps}} = - \frac{4m^{2}g^{2}}{(4\pi)^{4}\epsilon^{2}}\left[ 1 - \frac{1}{2}\epsilon - \epsilon\ln\left(\frac{m^{2}}{4\pi\mu^{2}}\right)\right]\,\mathbf{\Pi}^{2},
\end{eqnarray}
\begin{eqnarray}
\parbox{12mm}{\includegraphics[scale=1.0]{fig22.eps}} =  \frac{2m^{2}g^{2}}{(4\pi)^{4}\epsilon^{2}}\left[ 1 - \frac{1}{2}\epsilon - \frac{1}{2} \epsilon\ln\left(\frac{m^{2}}{4\pi\mu^{2}}\right)\right]\mathbf{\Pi}^{2},
\end{eqnarray}
\begin{eqnarray}
\parbox{12mm}{\includegraphics[scale=1.0]{fig23.eps}} =  \frac{6m^{2}g^{2}}{(4\pi)^{4}\epsilon^{2}}\left[ 1 - \frac{1}{2} \epsilon\ln\left(\frac{m^{2}}{4\pi\mu^{2}}\right)\right]\mathbf{\Pi}^{2},
\end{eqnarray}
\begin{eqnarray}
\left(\parbox{12mm}{\includegraphics[scale=1.0]{fig6.eps}}\right)\Biggr|_{P^{2} + K_{\mu\nu}P^{\mu}P^{\nu}=0} =   -\frac{6m^{2}g^{2}}{(4\pi)^{2}\epsilon}\left[ 1 + \frac{1}{2}\epsilon -\epsilon\ln\left(\frac{m^{2}}{4\pi\mu^{2}}\right)\right]\mathbf{\Pi}^{2},
\end{eqnarray}
where
\begin{eqnarray}\label{uhduhufgjg}
J(P^{2} + K_{\mu\nu}P^{\mu}P^{\nu}) = \int_{0}^{1}d x \ln \left[\frac{x(1-x)(P^{2}+ K_{\mu\nu}P^{\mu}P^{\nu}) + m^{2}}{4\pi\mu^{2}}\right],
\end{eqnarray}
\begin{eqnarray}\label{uhduhufgjgdhg}
&& J_{3}(P^{2} + K_{\mu\nu}P^{\mu}P^{\nu}) = \nonumber \\ && \int_{0}^{1}\int_{0}^{1}d x \,d y\,(1-y) \ln \Biggl\{\frac{y(1-y)(P^{2}+ K_{\mu\nu}P^{\mu}P^{\nu})}{4\pi\mu^{2}} +  \left[1-y + \frac{y}{x(1-x)}  \right]\frac{m^{2}}{4\pi\mu^{2}}\Biggr\},
\end{eqnarray}
\begin{eqnarray}\label{ugujjgdhg}
J_{4}(P^{2} + K_{\mu\nu}P^{\mu}P^{\nu}) =  \frac{m^{2}}{\mu^{2}}\int_{0}^{1}d x\frac{(1 - x)}{\frac{x(1 - x)(P^{2} + K_{\mu\nu}P^{\mu}P^{\nu})}{\mu^{2}} + \frac{m^{2}}{\mu^{2}}}.
\end{eqnarray}
The results for the diagrams above are in agreement with their counterparts in the limit of slight LV mechanism, \emph{i. e.}, the limit $| K_{\mu\nu}|\ll 1$ \citep{PhysRevD.84.065030,Carvalho2013850,Carvalho2014320}. We are now in a position to compute the $\beta$-function and anomalous dimensions. Then, we obtain
\begin{eqnarray}
\beta(g) = \frac{N + 8}{3(4\pi)^{2}}\mathbf{\Pi}g^{2} - \frac{3N + 14}{3(4\pi)^{4}}\mathbf{\Pi}^{2}g^{3},
\end{eqnarray}
\begin{eqnarray}
\gamma(g) = \frac{N + 2}{36(4\pi)^{4}}\mathbf{\Pi}^{2}g^{2} - \frac{(N + 2)(N + 8)}{432(4\pi)^{6}}\mathbf{\Pi}^{3}g^{3},
\end{eqnarray}
\begin{eqnarray}
\gamma_{m}(g) = \frac{N + 2}{6(4\pi)^{2}}\mathbf{\Pi}g - \frac{5(N + 2)}{36(4\pi)^{4}}\mathbf{\Pi}^{2}g^{2}.
\end{eqnarray}
We observe that from the two LV contributions of the theory, one of them has not survived in the final theory. It is that associated to external momenta. It has been canceled out in the renormalization process, as the BPHZ method demands. In fact, one of the main features of this renormalization scheme is its elegance and generality, where the momentum-dependent integrals cancel out order by order in perturbation theory in the renormalization program. Thus, we do not need to be evaluate these momentum-dependent integrals \cite{BogoliubovParasyuk,Hepp,Zimmermann}. In the present case, for the finite loop order inspected here, they are the ones in Eqs. (\ref{uhduhufgjg})-(\ref{ugujjgdhg}). The remaining LV contribution, that coming from the volume elements of Feynman diagrams, is the only responsible for LV modifications of the final theory. Ic occurs in the form of powers of the full LV $\mathbf{\Pi}$ factor. It is clear that the number of powers of $\mathbf{\Pi}$ in a given term of the $\beta$-function and anomalous dimensions is equal to just a single number, the number of loops of the corresponding term and not only to the number of powers of the dimensionless renormalized coupling constant through its rescaling, specially in the case of the $\beta$-function. On the other hand, the earlier non-exact approach \cite{PhysRevD.84.065030,Carvalho2014320,Carvalho2013850} is based on a couple of concepts, jointly the field and dimensionless renormalized coupling constant ones, where they are scaled by the few finite terms of $\mathbf{\Pi}$, namely $\Pi \simeq \Pi^{(0)} + \Pi^{(1)} + \Pi^{(2)}$. In this case, it is not so clear and direct, as in the exact approach, that in a given term of the $\beta$-function, the powers of $\Pi$ and dimensionless renormalized coupling constant are not the same. Thus, besides exact, our approach is simpler than the non-exact one and is based on a single concept, the loop number concept, and not on a couple of them, gives immediately and in a straightforward way, the final form of the $\beta$-function and anomalous dimensions expressions, as claimed in Section \ref{Introduction}. In next Section, we will use this fact, valid for next-to-leading loop level, for its generalization for any loop order.    

\section{Exact Lorentz-violating all-loop order quantum corrections}

\par We generalize the results of last Section to any loop level. Before that, we state a theorem needed in the generalization process. The present theorem is applicable only to the theory considered in this Letter because the referred Lorentz violation mechanism is manifest through the bilinear combination $P^{\mu}P^{\nu}$. This combination is distinct from that for a higher-derivative one $\frac{1}{2\Lambda_{L}^{2n-2}}(P^{2})^{n}$ \cite{PhysRevD.76.125011} for example. 

\begin{theorem} 
Consider a given Feynman diagram in momentum space of any loop order in a theory represented by the Lagrangian density of Eq. (\ref{bare Lagrangian density}). Its evaluated expression in dimensional regularization in $d = 4 - \epsilon$ can be written as a general functional $\mathbf{\Pi}^{L}\mathcal{F}(g,P^{2} + K_{\mu\nu}P^{\mu}P^{\nu},\epsilon,\mu,m)$ if its Lorentz-invariant (LI) counterpart is given by $\mathcal{F}(g,P^{2},\epsilon,\mu,m)$, where $L$ is the number of loops of the corresponding diagram.
\end{theorem}

\begin{proof} 
A general Feynman diagram of loop level $L$ is a multidimensional integral in $L$ distinct and independent momentum integration variables $q_{1}$, $q_{2}$,...,$q_{L}$, each one with volume element given by $d^{d}q_{i}$ ($i = 1, 2,...,L$). As showed in last Section, the substitution $q^{\prime} = \sqrt{\mathbb{I} + \mathbb{K}}\hspace{1mm}q$ transforms each volume element as $d^{d}q^{\prime} = \sqrt{\det(\mathbb{I} + \mathbb{K})}\,d^{d}q$. Thus $d^{d}q = d^{d}q^{\prime}/\sqrt{\det(\mathbb{I} + \mathbb{K})} \equiv \mathbf{\Pi}\,d^{d}q^{\prime}$, $\mathbf{\Pi} = 1/\sqrt{\det(\mathbb{I} + \mathbb{K})}$. Then, the integration in $L$ variables results in a LV overall factor of $\mathbf{\Pi}^{L}$. Now making $q^{\prime} \rightarrow P^{\prime}$ in the substitution above, where $P^{\prime}$ is the transformed external momentum, then $P^{\prime 2} = P^{2} + K_{\mu\nu}P^{\mu}P^{\nu}$. So a given Feynman diagram, evaluated in dimensional regularization in $d = 4 - \epsilon$, assumes the expression $\mathbf{\Pi}^{L}\mathcal{F}(g,P^{2} + K_{\mu\nu}P^{\mu}P^{\nu},\epsilon,\mu,m)$, where $\mathcal{F}$ is associated to the corresponding diagram if the LI Feynman diagram counterpart evaluation results in $\mathcal{F}(g,P^{2},\epsilon,\mu,m)$. 
\end{proof} 

\par Now by applying the BPHZ method, for all loop orders, in which all momentum-dependent integrals are eliminated in the renormalization process, order by order in perturbation theory \cite{BogoliubovParasyuk,Hepp,Zimmermann} and the theorem above, a possible LV dependence on $K_{\mu\nu}$ of $\beta$-function and anomalous dimensions coming from the LV momentum-dependent integrals disappears. So, according to the theorem aforementioned, the only LV dependence of these functions on the LV parameter $K_{\mu\nu}$ is the remaining one, that coming from the volume elements of the Feynman integrals, \textit{i. e.}, through the LV full $\mathbf{\Pi}^{L}$ factor, where $L$ is the number of loops of the referred term for these functions. Thus, according to the argument just exposed, we can write the exact LV radiative quantum corrections to the $\beta$-function and anomalous dimensions for all loop orders as
\begin{eqnarray}\label{uhgufhduhufdhu}
\beta(g) = \sum_{L=2}^{\infty}\beta_{L}^{(0)}\mathbf{\Pi}^{L-1}g^{L}, 
\end{eqnarray}
\begin{eqnarray}
\gamma(g) = \sum_{L=2}^{\infty}\gamma_{L}^{(0)}\mathbf{\Pi}^{L}g^{L},
\end{eqnarray}
\begin{eqnarray}
\gamma_{m}(g) = \sum_{L=1}^{\infty}\gamma_{m,L}^{(0)}\mathbf{\Pi}^{L}g^{L},
\end{eqnarray}
where $\beta_{L}^{(0)}$, $\gamma_{L}^{(0)}$ and $\gamma_{m,L}^{(0)}$ are the respective LI $L$th-loop radiative quantum corrections to the corresponding functions. When we are faced with the task of computing radiative quantum corrections in LV theories, the question of possible fine-tuning emerges. This question is avoided when we perform the renormalization program \cite{REYES2015190}, since this program fixes the physical quantities of a field theory through the renormalization group flow \cite{Maggiore:845116}. According to what was presented in Ref. \cite{Kleinert} for a $\lambda\phi^{4}$ theory which is corresponding to the studied here, the radiative corrections remain small, at least at five-loop order. We can see, one more time, even extending our considerations for all-loop level, that the exact approach is capable of furnishing, in a simple and straightforward way, the all-loop LV radiative quantum corrections to the $\beta$-function and anomalous dimensions in terms of a single concept, that of the loop number of the corresponding loop term. These results for the all-loop LV radiative quantum corrections would must be obtained from their non-perturbative counterparts in the appropriate limit, \emph{i.e.} in the weak coupling regime. As a concrete example, the known higher-loop analytically computed LV result is the LV three-loop radiative quantum correction to field anomalous dimension \cite{Carvalho2014320}. Following the prescription that each loop term of $\gamma(g)$ acquires a full LV $\mathbf{\Pi}$ factor, its three-loop result can be improved to its four-loop one \cite{PhysRevD.96.036016}, namely $\gamma(g) = \frac{(N+2)\mathbf{\Pi}^{2}g^{2}}{36(4\pi)^{4}} - \frac{(N+2)(N+8)\mathbf{\Pi}^{3}g^{3}}{432(4\pi)^{6}} - \frac{5(N^{2}-18N-100)(N+2)\mathbf{\Pi}^{4}g^{4}}{5184(4\pi)^{8}}$. The same prescription can be applied up to the six-loop \cite{PhysRevD.96.036016} radiative quantum corrections to field anomalous dimension as well as to $\beta$-function and mass anomalous dimension. In practice, the full LV $\mathbf{\Pi}$ factor can be computed by using its components whose bounds have been provided in Ref. \cite{PhysRevD.83.11}.    

\section{Asymptotic single-particle states and the LSZ formula}

\par We now investigate how the LV mechanism modifies the structure of the asymptotic single-particle states \cite{PhysRevD.90.065003}. For that, we can define ``in" and ``out" states represented by the fields $\phi_{in}$ and $\phi_{out}$, respectively. The fields satisfy to the LV Klein-Gordon equation
\begin{eqnarray}
(\partial_{\mu}\partial^{\mu} + K_{\mu\nu}\partial^{\mu}\partial^{\nu} - m^{2})\phi_{in, out} = 0.
\end{eqnarray}
The fields can be written in terms of plane-wave basis $u_{\textbf{p}}(\textbf{x}, t)$, for example for $\phi_{\textbf{p},in}$, as
\begin{eqnarray}
\phi_{\textbf{p},in}(\textbf{x}, t) = \int d^{3}p [a_{\textbf{p},in}u_{\textbf{p}}(\textbf{x}, t) + a_{\textbf{p},in}^{\dagger}u_{\textbf{p}}^{\star}(\textbf{x}, t)],
\end{eqnarray}
where $a_{\textbf{p},in}$ and $a_{\textbf{p},in}^{\dagger}$ are the destruction and creation operators of the field $\phi_{\textbf{p},in}$ which satisfy the usual commutation relation for boson fields. Thus we can define the $n$-particles Fock-space momentum states 
\begin{eqnarray}
\ket{\textbf{p}_{1},...,\textbf{p}_{n};in} = a_{\textbf{p}_{1},in}^{\dagger}\cdots a_{\textbf{p}_{n},in}^{\dagger}\ket{0},
\end{eqnarray}
where $\ket{0}$ is the vacuum of the theory. Defining analog ``out" states
\begin{eqnarray}
\ket{\textbf{q}_{1},...,\textbf{q}_{n};out} = a_{\textbf{q}_{1},out}^{\dagger}\cdots a_{\textbf{q}_{n},out}^{\dagger}\ket{0},
\end{eqnarray}
we can define the amplitudes for a transition from the $n$-particle ``in" state to the $m$-particle ``out" one by
\begin{eqnarray}
S_{fi} = \bra{\textbf{q}_{1},...,\textbf{q}_{m};out}\ket{\textbf{p}_{1},...,\textbf{p}_{n};in},
\end{eqnarray}
where $S_{fi}$ are the matrix elements of the $S$ matrix. Now we can express the $S_{fi}$ elements in terms of the $n$-particle Green's function $G^{(n)}(\textbf{p}_{1},...,\textbf{p}_{n})$. This is possible due to the LSZ reduction formula
\begin{eqnarray}
&& S_{fi} = \left(\frac{-i}{\sqrt{Z}}\right)^{n + m}N_{\textbf{q}_{1}}\cdots N_{\textbf{q}_{m}}N_{\textbf{p}_{1}}\cdots N_{\textbf{p}_{m}}(q_{1}^{2} + K_{\mu\nu}q_{1}^{\mu}q_{1}^{\nu} - m^{2})\cdots (q_{m}^{2} + K_{\mu\nu}q_{m}^{\mu}q_{m}^{\nu} - m^{2})\times \nonumber \\ && (p_{1}^{2} + K_{\mu\nu}p_{1}^{\mu}p_{1}^{\nu} - m^{2})\cdots (p_{n}^{2} + K_{\mu\nu}p_{n}^{\mu}p_{n}^{\nu} - m^{2})G^{(n + m)}(\textbf{q}_{1},...,\textbf{q}_{m},-{p}_{1},...,-\textbf{p}_{n}),
\end{eqnarray}
where $Z$ is a renormalization constant which relates the interacting field $\phi$ to the ``in" and ``out" ones through
\begin{eqnarray}
\lim_{x_{0}\rightarrow  -\infty} \phi(x) = \sqrt{Z}\phi(x)_{in},
\end{eqnarray}
\begin{eqnarray}
\lim_{x_{0}\rightarrow  \infty} \phi(x) = \sqrt{Z}\phi(x)_{out},
\end{eqnarray}
$N_{\textbf{p}} = ((2\pi)^{3}2\omega_{\textbf{p}})^{-1/2}$.

\section{Conclusions}

\par In this Letter, we presented the all-loop renormalization of the LV O($N$) $\lambda\phi^{4}$ scalar field theory, where the Lorentz violation mechanism is treated exactly by keeping the LV $K_{\mu\nu}$ coefficients in their nonperturbative form. This was achieved firstly through an analytic and explicit computation, at next-to-leading order, of the exact LV radiative quantum corrections to the $\beta$-function and anomalous dimensions. The full task was completed after we extended the formed employment, valid for all loop levels, through an induction process based on a general theorem emerging from the exact approach. The exact approach can, in a simple and straightforward way, furnish expressions for the exact all-loop LV radiative quantum corrections to the $\beta$-function and anomalous dimensions considering just the single concept of loop number, the loop number of the corresponding term of these functions, thus being, besides exact, simpler than the earlier non-exact approach. We showed that the exact results reduced to its non-exact, earlier obtained, counterparts in the appropriated limit. The exact calculation presented in this Letter, involving such a symmetry breaking mechanism in the referred theory, is the first one in literature for our knowledge. The LV symmetry breaking mechanism approached here is a naive one. Our approach contributes for the implementation of a renormalization procedure in theories with others Lorentz symmetry breaking mechanisms where the LV factor can be momentum-dependent \cite{S0217751X17500841}, which would avoid fine-tuning problems in these theories. In fact, applications of its possible effects on phase transitions and critical phenomena at $d<4$ were investigated \cite{EurophysLett10821001,doi:10.1142/S0217979215502598,doi:10.1142/S0219887816500493} through the computation of the radiative quantum corrections to critical exponents. As the critical exponents are universal quantities, they depend on the properties of the field such as its symmetry. As the Lorentz violation symmetry breaking mechanism occurs in the space-time where the field is defined and not in its internal one, the field symmetry is not changed and the critical exponents are the same as that obtained in the Lorentz-invariant theory, although the non-universal quantities as the renormalization constants, $\beta$-functions, anomalous dimensions, fixed point etc depend on $K_{\mu\nu}$ both explicitly through momentum-dependent parametric integrals which cancel out in the renormalization program and implicitly through the full LV $\mathbf{\Pi}$ factor. The computation of the final results for the critical exponents has shown the canceling of the $\Pi$ factor (the limit of the full LV $\mathbf{\Pi}$ factor for small $K_{\mu\nu}$ values) and the critical exponents do not depend on $K_{\mu\nu}$. This shows that the referred symmetry braking mechanism is naive. Thus, when dealing with universal quantities as critical exponents, we have to verify through the entire renormalization program that the effect of $K_{\mu\nu}$ is null. That was what was confirmed in the works \cite{EurophysLett10821001,doi:10.1142/S0217979215502598,doi:10.1142/S0219887816500493}. This occurs because the LV theory is related to its Lorentz-invariant counterpart through a redefinition of the metric. Furthermore, the exact approach applied here can inspire the task of considering the exact effect of LV mechanisms in distinct theories in high energy physics field theories of standard model extension for example as well as in low energy physics for condensed matter field theories in considering the problems of corrections to scaling, finite-size scaling, amplitude ratios as well as critical exponents in geometries subjected to different boundary conditions for systems belonging to the O($N$) and Lifshitz \cite{PhysRevB.67.104415,PhysRevB.72.224432,Carvalho2009178,Carvalho2010151} universality classes etc. 

\section*{Acknowledgements}

\par With great pleasure the authors thank the kind referee for helpful comments. PRSC and MISJ would like to thank Federal University of Piau\'{i} and FAPEAL (Alagoas State Research Foundation), CNPq (Brazilian Funding Agency) for financial support, respectively.

\bibliography{apstemplate}

\end{document}